\documentclass[12pt]{article}
\usepackage[T1]{fontenc}
\usepackage[utf8]{inputenc}
\usepackage{amsmath,amssymb,amscd,latexsym}
\usepackage{graphicx}
\usepackage{epsfig}
\usepackage{fancyhdr}
\usepackage{mathtools}
\usepackage{amsthm,comment}
\usepackage[dvipsnames]{xcolor}
\usepackage{framed, color, thmtools}
\usepackage{wasysym}
\usepackage{authblk}
\usepackage[margin=2cm]{geometry}

\newcommand{\N}{\mathbb{N}}
\newcommand{\R}{\mathbb{R}}

\newcommand{\C}{\mathcal{C}}

\newcommand{\F}{\mathbb{F}}
\newcommand{\A}{\mathcal{A}}

\newcommand{\Cd}{\mathcal{C}^\perp}

\newcommand{\mat}{\F_q^{ \, n \times m }}

\newcommand{\maxrk}{\mbox{maxrk}}
\newcommand{\dmin}{\mbox{d}}
\newcommand{\rank}{\mbox{rank}}
\newcommand{\rowsp}{\mbox{rowsp}}
\newcommand{\colsp}{\mbox{colsp}}

\newcommand{\Mat}{\mbox{Mat}}

\theoremstyle{definition}
\newtheorem{example}{Example}[section]
\newtheorem{definition}[example]{Definition}

\newtheorem{proposition}[example]{Proposition}
\newtheorem{lemma}[example]{Lemma}
\newtheorem{theorem}[example]{Theorem}

\newtheorem{remark}[example]{Remark}
\newtheorem{notation}[example]{Notation}

\declaretheoremstyle[
mdframed={
        backgroundcolor=lightgray,
        linecolor=lightgray}
]{colored}

\declaretheoremstyle[
 mdframed={
        backgroundcolor=lightgray,
        linecolor=lightgray}
]{colored}

\makeindex

\begin{document}

\title{Quasi optimal anticodes: structure and invariants}
\author{Elisa Gorla and Cristina Landolina}
\affil{Institut de Math\'ematiques, Universit\'e de Neuch\^atel, Switzerland}
\date{}

\maketitle

\abstract 
It is well-known that the dimension of optimal anticodes in the rank-metric is divisible by the maximum $m$ between the number of rows and columns of the matrices. Moreover, for a fixed $k$ divisible by $m$, optimal rank-metric anticodes are the codes with least maximum rank, among those of dimension $k$. In this paper, we study the family of rank-metric codes whose dimension is not divisible by $m$ and whose maximum rank is the least possible for codes of that dimension, according to the Anticode bound. As these are not optimal anticodes, we call them quasi optimal anticodes (qOACs). In addition, we call dually qOAC a qOAC whose dual is also a qOAC. We describe explicitly the structure of dually qOACs and compute their weight distributions, generalized weights, and associated $q$-polymatroids.

\section*{Introduction and motivation}

Rank-metric codes are linear spaces of matrices of given size, with entries in a finite field $\mathbb{F}_q$. The set of $n\times m$ matrices over a field is a metric space when endowed with the rank metric: The distance between two matrices is the rank of their difference. Rank-metric codes were introduced by Delsarte in \cite{Del}. Gabidulin \cite{Gab} and Roth \cite{Roth} later independently rediscovered the family of rank-metric codes, which are linear over the field extension $\mathbb{F}_{q^m}$.  
The study of rank-metric codes is motivated by their applications to linear network coding, code-based cryptography, and distributed storage. 

The mathematical theory of rank-metric codes includes the study of code invariants and their bounds. Beyond the size of its elements, also called codewords, and its dimension, two basic invariants of a rank-metric code are its minimum distance and maximum rank. The minimum distance of a nonzero rank-metric code is the least rank of a nonzero matrix in the code, while its  maximum rank is the maximum rank of an element of the code.
The Singleton bound involves the dimension of a rank-metric code and its minimum distance, while the Anticode bound involves the dimension of a code and its maximum rank. Codes which meet the Singleton bound go under the name of Maximum Rank Distance codes (MRD codes) and were first studied by Delsarte in \cite{Del}. Codes attaining the Anticode bound were studied by Meshulam in \cite{MES}, with a different motivation and terminology. Within coding theory, they are known as optimal anticodes. While MRD codes are of direct applied interest, optimal anticodes are mostly of theoretical interest. For example, they can be used to define the generalized weights, a family of code invariants. This was done in \cite{ACapproach}, where the properties of optimal anticodes are studied.

For a fixed minimum distance $d$, MRD codes are the rank-metric codes which have minimum distance $d$ and the largest dimension according to the Singleton bound. Conversely, one could fix a dimension $k$ and ask what are the codes of dimension $k$ and largest minimum distance. It turns out that, if $k$ is divisible by the maximum between $n$ and $m$, these are the MRD codes. Else, they are a different family of codes, called  quasi MRD codes. This last family of codes was introduced and studied in \cite{qMRD}.

Something similar happens for the Anticode bound: The optimal anticodes are the rank-metric codes of largest dimension, among those with a fixed maximum rank. Similarly to the case of MRD codes, the dimension of optimal anticodes is divisible by the maximum between $n$ and $m$, say $m$. Moreover, for a fixed $k$ divisible by $m$, optimal anticodes are the codes with least maximum rank, among those of dimension $k$. In this paper, we study the family of rank-metric codes whose dimension is not divisible by $m$ and whose maximum rank is the least possible for codes of that dimension, according to the Anticode bound. As these are not optimal anticodes, we call them quasi optimal anticodes (qOACs).

The paper is organized as follows: In Section \ref{SecqOACs} we define qOACs and dually qOACs. We give a complete structural classification of dually qOACs and a partial one of qOACs. In Section \ref{Secgenweights} we study the generalized weights of dually qOACs and of a family of qOACs. In Section \ref{SecRankdistr} we study their the weight distribution, while in Section \ref{SecPolymat} we compute the associated $q$-polymatroids.

\section{Preliminaries}

In this section we give an introduction to rank-metric codes and state some results that we will need throughout the paper. 
Let $q$ a prime power and $\mat$ the linear space of $n \times m$ matrices with entries in the finite field $\F_q$. Up to transposition, we may assume without loss of generality that $n \leq m$.
Throughout the paper, we denote by $e_1,\ldots,e_\ell$ the canonical basis of $\F_q^\ell$, i.e., $e_i$ is the vector whose coordinates are all zero, except for a one in position $i$.

\begin{definition}
The \textbf{rank distance} of $A, B \in \mat$ is given by the function 
 \begin{align*}
 d :  \mat  \times \mat & \longrightarrow  \N \\
      (A,B)     \qquad        & \longmapsto      \rank(A-B). \\
\end{align*}
A \textbf{rank-metric code} $\mathcal{C}$ is a $\F_q$-linear subspace of $\mat$ equipped with the rank distance. The \textbf{dimension} of $\mathcal{C}$ over $\F_q$ is denoted by $\dim(\mathcal{C})$. Further, we let $\dmin(\mathcal{C}) = \min \{ \rank(A) : A \in \mathcal{C} , A \neq 0 \}$ be the \textbf{minimum distance} of a nonzero rank-metric code $\mathcal{C}$. In a similar way, the \textbf{maximum rank} is given by $\maxrk(\mathcal{C}) = \max \{ \rank(A) : A \in \mathcal{C}\}. $ 
\end{definition}

Two rank-metric codes are equivalent if there is a linear rank-preserving homomorphism mapping one code into the other. 

\begin{definition}
An \textbf{$\F_q$-linear isometry} $\varphi$ of $\mat$ is an $\F_q$-linear rank-preserving homomorphism, i.e. $\rank(\varphi(A)) = \rank(A)$ for all $A \in \mat$. Let $\mbox{Isom}_{\F_q}(\mat)$ be the collection of $\F_q$-linear isometries of $\mat$. Two rank-metric codes $\mathcal{C}, \mathcal{D} \subseteq \mat$ are \textbf{equivalent} if there is $\varphi \in \mbox{Isom}_{\F_q}(\mat)$ such that $\mathcal{C} = \varphi(\mathcal{D})$. We denote the equivalence by $\mathcal{C} \sim \mathcal{D}$.
\end{definition}

Linear isometries of $\mat$ are classified in \cite{Hua} for odd characteristic and in \cite{Wan} for characteristic equal to 2. 

\begin{theorem}
Let $\varphi \in \mbox{Isom}_{\F_q}(\mat)$. Then 
\begin{itemize}
    \item[(a)] if $n \neq m$, there exist $N \in \mbox{GL}_n(\F_q)$ and $M \in \mbox{GL}_m(\F_q)$ such that $\varphi(A) = NAM$ for all $A \in \mat$.
    \item[(b)] if $n = m$, there exist $N,M \in \mbox{GL}_n(\F_q)$ such that $\varphi(A) = NAM$ for all $A \in \mat$, or $\varphi(A) = NA^tM$ for all $A \in \mat$.
\end{itemize}
\end{theorem}

Further we define the dual of a rank-metric code $\mathcal{C}$ using the standard scalar product for matrices.

\begin{definition}
Let $\mathcal{C} \subseteq \mat$ be a rank-metric code and let $\mbox{Tr}(A)$ denote the trace of a square matrix $A$. The \textbf{dual} of $\mathcal{C}$ is defined as
$$\Cd = \{M \in \mat : \mbox{Tr}(MN^t) = 0 \mbox{ for all } N \in \mathcal{C}  \}.$$
\end{definition}

For a matrix $M\in \mat$  $\colsp(M)\subseteq \F_q^n$ denotes the $\F_q$-linear space generated by the columns of $M$ and $\rowsp(M) \subseteq \F_q^m$  the $\F_q$-linear space generated by the rows of $M$.

\begin{definition}\label{supportspace}
Let $V \subseteq \F_q^n$ be a subspace. Then
$$\Mat(V) = \{ M \in \mat : \colsp(M) \subseteq V \} \subseteq \mat$$
denotes the \textbf{matrix space supported by the vector space $V$}. For $\mathcal{C} \subseteq \mat$ a rank-metric code, let
$$\mathcal{C}(V) = \{ M \in \mathcal{C} : \colsp(M) \subseteq V \} \subseteq \mathcal{C}$$
be the \textbf{subcode of $\mathcal{C}$ supported on $V$}.
\end{definition}

Notice that whenever we want to refer to the matrix space supported by a certain row space we will consider the transposed version of the support given in Definition \ref{supportspace}. 
\medskip

An important role in the motivation of this paper is taken by deriving bounds on rank-metric codes. In the sequel we give two of the most relevant inequalities, one relating the minimum distance and one relating the maximum rank to the dimension of a rank-metric code. 

\begin{theorem}(Singleton bound)
Let $\mathcal{C} \subseteq \mat$ be a rank-metric code. Then
\begin{equation} \label{SB}
\dim(\mathcal{C}) \leq m (n - \dmin(\mathcal{C}) + 1).
\end{equation}
\end{theorem}
The \textbf{Singleton-like bound} for rank-metric codes was presented by Delsarte in \cite{Del}. It is the rank-metric version of the well-known Singleton bound in the Hamming metric. Codes meeting bound (\ref{SB}) are known as \textbf{Maximum Rank Distance Codes (MRD codes)} and have been extensively studied. 

We shall present now an upper bound on the dimension involving the maximum rank, instead of the minimum distance as seen in (\ref{SB}).  
A classical theorem by Flanders in \cite{FLA} states that the dimension of a linear space of matrices whose rank is less than or equal to a given $r \leq n$ is upper bounded by $mr$.  The results in \cite{FLA} are proved under the assumption that the cardinality of the base field is strictly greater than $r$ and that the characteristic differs from 2. Atkinson and Lloyd in \cite{AL} obtained the same result with the assumption only on the field size.  The square case with $r = n-1$ for an arbitrary field size was proved by Dieudonné in \cite{DIE}. Finally, Meshulam in \cite{MES} showed that the assumptions on the rank and on the field size are unnecessary for deriving the bound on the dimension. In fact, the next bound was proved in \cite{MES} and goes under the name of \textbf{Anticode bound}.

\begin{theorem}(Anticode bound) 
Let $\mathcal{C} \subseteq \mat$ be a rank-metric code. Then 
\begin{equation}\label{ACB}
\dim(\mathcal{C}) \leq m \, \maxrk(\mathcal{C}). 
\end{equation}
\end{theorem}

\begin{definition}
An \textbf{optimal anticode} $\A \subseteq \mat$ is a rank-metric code which satisfies
$$\dim(\A) = m \,  \maxrk(\A).$$
\end{definition}

The classification of matrix spaces with least possible maximum rank for a given dimension follows the same history as the derivation of the Anticode Bound and is presented in the next theorem.

\begin{theorem}
Let $\A \subseteq \mat$ be an optimal anticode of dimension $mr$ with $r = \maxrk(\A)$. Then $\A=\Mat(V)$ for some $V \subseteq \F_q^n$ of dimension $r$, or $A=\Mat(V)^t$ and $m=n$.
\end{theorem}

Some fundamental properties of optimal anticodes are stated in \cite{RAV}, for instance that $\mathcal{C}$ is an optimal anticode if and only if $\Cd$ is an optimal anticode.

In the end we want to introduce $q$-polymatroids, which are the $q$-analogue of polymatroids. 

\begin{definition}[\cite{q-poly}, Definition 4.1]
A \textbf{$q$-polymatroid} is a pair $P = (\F_q^n , \rho)$ where $\rho$ is a function from the set of all subspaces of $\F_q^n$ to $\R$, such that 
\begin{itemize}
\item[(i)] $0 \leq \rho(A) \leq \dim(A)$,
\item[(ii)] if $A \subseteq B$, then $\rho(A) \leq \rho(B)$,
\item[(iii)]$ \rho(A + B) + \rho(A \cap B) \leq \rho(A) + \rho(B)$.
\end{itemize}
for all $A, B \subseteq \F_q^n$.
\end{definition}

\notation{Let $s,t$ be positive integers. Throughout the paper we let $u \in \F_q^s$ denote a row vector of length $s$ with entries in $\F_q$. The zero matrix of size $s \times t$ is denoted by  $0_{s\times t}$ and the matrix with only one nonzero entry equal to one at position $(i,j)$ of size $n \times m$ is given by $E_{i,j}$.
Further, $f_i \in \F_q^m$ is the row vector with the only non zero entry equal 1 at position $i$, whereas $e_i \in \F_q^n$ is the column vector whose only non zero entry is a 1 in position $i$.}

\section{Quasi optimal and dually quasi optimal anticodes}\label{SecqOACs}

\begin{definition}\label{def} 
A \textbf{quasi optimal anticode (qOAC)} is a rank-metric code $\mathcal{C}\subseteq\mat$ such that $m\nmid\dim(\mathcal{C})$ and  
$$\maxrk(\mathcal{C})=\bigg \lceil\frac{\dim(\mathcal{C})}{m}\bigg\rceil.$$ 
If $\mathcal{C}$ and $\Cd$ are both qOACs, then $\mathcal{C}$ is a {\bf dually qOAC}. 
\end{definition}

\begin{notation}
For a qOAC $\mathcal{C}$, write 
$$\dim(\mathcal{C}) = \alpha m + \rho,\ \ 0 < \rho < m,\ \ 0\leq\alpha<n\leq m,\ \ \maxrk(\mathcal{C})=\alpha+1.$$
If in addition $\mathcal{C}$ is a dually qOAC, write
$$\dim(\mathcal{C}^\perp) = (n - \alpha - 1) m + (m - \rho),\ \ 0 < m -\rho < m,\  \maxrk(\mathcal{C}^\perp) = n - \alpha.$$
\end{notation}

It is well-known that the dual of an optimal anticode is an optimal anticode. However, this is not the case for qOACs. In the next example, we produce dually qOACs, as well as qOACs which are not dually qOACs.

\begin{example}\label{ex:duals}
Let $m\geq\max\{2,n\}$, let $0 < \alpha<n\leq m$, $0<\rho<m$, $0\leq k\leq m-\rho$. Let
$$\mathcal{C}_k=\left\{\begin{pmatrix}
 A & B \\
u & 0_{1\times k} \\
w & 0_{1 \times (m-\rho-k)} \\
  0_{(n-\alpha-1) \times  (\rho+k)} & 0_{(n-\alpha-1) \times (m-\rho - k)} \\
\end{pmatrix} : A \in \F_q^{(\alpha - 1) \times (m-k)}, B \in \F_q^{(\alpha - 1) \times k}, u \in \F_q^{m-k} , w \in \F_q^{\rho + k} \right\}$$ 
with dual code
$$\Cd_k=\left\{\begin{pmatrix}
 0_{(\alpha - 1) \times (m-k)} & 0_{(\alpha - 1) \times k} \\
0_{1 \times (m-k)}  & u \\
0_{1 \times (\rho + k)} & w \\
  A & B \\
\end{pmatrix} : A \in \F_q^{(n-\alpha - 1) \times (\rho + k)}, B \in \F_q^{(n- \alpha - 1) \times (m-\rho-k)}, u \in \F_q^{k} , w \in \F_q^{m- \rho - k} \right\}.$$
Notice that $\mathcal{C}_k \sim \mathcal{C}_{m-\rho-k}$ for all $k$. 
We have $\maxrk(\mathcal{C}_0) = \alpha + 1$ and $\maxrk(\Cd_0) = n - \alpha$, hence $\mathcal{C}_0\sim \mathcal{C}_{m-\rho}$ are dually qOACs. If $\rho\geq m-2$, then $\maxrk(\mathcal{C}_1) = \alpha + 1$ and $\maxrk(\Cd_1) = n - \alpha$, hence $\mathcal{C}_1\sim \mathcal{C}_{m-\rho-1}$ are dually qOACs. If $\rho\leq m-3$, then $\maxrk(\mathcal{C}_1)=\alpha+1$ and $\maxrk(\mathcal{C}_1^\perp)=n-\alpha+1$. Hence $\mathcal{C}_1\sim \mathcal{C}_{m-\rho-1}$ are qOACs, but not dually qOACs. For $k\neq 0,1,m-\rho-1,m-\rho$, one has $\maxrk(\mathcal{C}_k) = \alpha + 1$ and $\maxrk(\Cd_k) = n - \alpha + 1$. Therefore $\mathcal{C}_k\sim \mathcal{C}_{m-\rho-k}$ are qOACs, but not dually qOACs.
\end{example}

The next proposition relates the maximum rank of a code with that of its dual. It also provides us with a simple characterization of dually qOACs.

\begin{proposition}\label{prop:n+1}
Let $\mathcal{C}\subseteq\mat$ be a rank-metric code. Then $$\maxrk(\mathcal{C})+\maxrk(\mathcal{C}^\perp)\geq n.$$
In addition:
\begin{enumerate}
\item[(a)] $\mathcal{C}$ is an optimal anticode if and only if $\maxrk(\mathcal{C})+\maxrk(\mathcal{C}^\perp)=n$.
\item[(b)] $\mathcal{C}$ is a dually qOAC if and only if  $\maxrk(\mathcal{C})+\maxrk(\mathcal{C}^\perp)= n + 1$.
\end{enumerate}
\end{proposition} 

\begin{proof}
The Anticode Bound on $\mathcal{C}$ and $\mathcal{C}^\perp$ yields \begin{equation}\label{eqn:sumofmaxrk}
mn=\dim(\mathcal{C})+\dim(\mathcal{C}^\perp)\leq m(\maxrk(\mathcal{C})+\maxrk(\mathcal{C}^\perp)).
\end{equation}

(a) $\mathcal{C}$ is an optimal anticode if and only if $\mathcal{C}^\perp$ is an optimal anticode. Hence, if $\mathcal{C}$ is an optimal anticode, then equality holds in~(\ref{eqn:sumofmaxrk}). Conversely, if equality holds in~(\ref{eqn:sumofmaxrk}), then $\mathcal{C}$ and $\mathcal{C}^\perp$ meet the Anticode Bound, hence they are optimal anticodes.

(b) If $\mathcal{C}$ is a dually qOAC, then $\maxrk(\mathcal{C})+\maxrk(\mathcal{C}^\perp)=n+1$ by direct computation. To prove the converse, first observe that, if $\maxrk(\mathcal{C})+\maxrk(\mathcal{C}^\perp)=n+1$, then $\C$ and $\C^\perp$ are not optimal anticodes by part (a).
If $m\mid \dim(\mathcal{C})$, then
$$\maxrk(\mathcal{C})+\maxrk(\mathcal{C}^\perp)\geq \frac{\dim(\mathcal{C})}{m}+1+\frac{\dim(\mathcal{C}^\perp)}{m} +1=n+2,$$ since $\C$ and $\C^\perp$ are not optimal anticodes. If instead $m\nmid\dim(\C)$, then \begin{equation}\label{eqn:n+1}
n+1=\maxrk(\mathcal{C})+\maxrk(\mathcal{C}^\perp)\geq \left\lceil\frac{\dim(\mathcal{C})}{m}\right\rceil+\left\lceil\frac{\dim(\mathcal{C}^\perp)}{m}\right\rceil\geq\alpha+1+n-\alpha,
\end{equation} where $\dim(\C)=\alpha m+\rho$, $\dim(\C^\perp)=(n-\alpha-1)m+(m-\rho)$, and $\rho>0$. Therefore the inequalities in (\ref{eqn:n+1}) are equalities, which completes the proof.
\end{proof}

A first classification of large matrix spaces of bounded rank appears in \cite{AL}, where Atkinson and Lloyd study linear spaces of dimension close to $mr$ over fields of large cardinality. Their classification was extended to all fields and matrices of arbitrary size by de Seguins Pazzis in \cite{Seguins}. As a direct consequence of the results by de Seguins Pazzis, we can characterize the qOACs whose dimension is at least $\alpha(m-1)+n$. 
\begin{theorem}[\cite{Seguins}, Theorem 4, Theorem 5, and Theorem 6] \label{firstclasstheorem}
Let $\mathcal{C} \subseteq \mat$ be a qOAC of $\dim(\mathcal{C}) = \alpha m + \rho$, $0<\rho<m$. 
\begin{enumerate}
\item[(a)] If $\rho>n-\alpha$, then $\mathcal{C}$ is equivalent to a linear subspace of $\Mat(\langle e_1,\ldots,e_{\alpha+1}\rangle)$.
\item[(b)] If $\rho = n - \alpha$, then one of the following holds:
\begin{itemize}
\item[(i)] $\mathcal{C}$ is equivalent to a linear subspace of  $\Mat(\langle e_1,\ldots,e_{\alpha+1}\rangle)$,
\item[(ii)] $\mathcal{C} \sim \Mat(\langle e_1,\ldots,e_{\alpha}\rangle)+\Mat(\langle e_1\rangle)^t$, \\ 
\item[(iii)] $m = n+1$ and
$\mathcal{C} \sim \Mat(\langle e_1,\ldots,e_{\alpha+1}\rangle)^t$,
\item[(iv)] $m=n=3$, $q=2$, and $\mathcal{C}\sim\left\{\begin{pmatrix}
a & 0 & 0 \\ c & b & 0 \\ d & e & a+b
\end{pmatrix} : (a,b,c,d,e)\in\mathbb{F}_2^5 \right\}$.
\end{itemize}
\end{enumerate}
\end{theorem}

Using Theorem~\ref{firstclasstheorem}, we can classify dually qOAs.  

\begin{theorem}\label{corollaryC}
Let $\mathcal{C} \subseteq \mat$ be a dually qOAC with $\dim(\mathcal{C}) = \alpha m + \rho$, $0<\rho<m$. One of the following holds:
\begin{enumerate}
\item[(a)] $$\mathcal{C} \sim \left\{\begin{pmatrix}
A & B\\
u & 0_{1 \times (m-\rho)} \\
0_{(n-\alpha -1) \times \rho} & 0_{(n-\alpha -1) \times (m-\rho)} \\
\end{pmatrix} : A\in\F_q^{\alpha\times\rho}, B\in\F_q^{\alpha\times(m-\rho)}, u\in\F_q^{1\times \rho}\right\},$$
\item[(b)] $\rho\leq n-\alpha$ 
and $$\mathcal{C} \sim \left\{ \begin{pmatrix}
u & A \\
v & 0_{\rho\times(m-1)}\\
0_{(n-\alpha-\rho)\times 1} & 0_{(n-\alpha-\rho)\times(m-1)}\\
\end{pmatrix} : A\in \F_q^{\alpha\times(m-1)}, \\ 
u\in \F_q^{\alpha\times 1}, v\in \F_q^{\rho\times 1}\right\},$$
\item[(c)] $\rho\geq m-\alpha-1$ 
and 
\begin{equation*}
\mathcal{C}\sim\left\{\begin{pmatrix}
A & u \\
B & 0_{(m-\rho)\times 1}\\
0_{(n-\alpha-1)\times (m-1)} & 0_{(n-\alpha-1)\times 1}\\
\end{pmatrix} : A\in \F_q^{(\alpha+\rho+1-m)\times(m-1)},
B\in \F_q^{(m-\rho)\times(m-1)}, u\in \F_q^{(\alpha+\rho+1-m)\times 1}\right\}.
\end{equation*}
\item[(d)] $m=n+1$, $\rho=n-\alpha$, and 
$$\mathcal{C} \sim \left\{ \begin{pmatrix}
A & 0_{n \times (n-\alpha)} \\
\end{pmatrix} : A \in \F_q^{n \times (\alpha +1)} \right\}.$$
\end{enumerate}
\end{theorem}

\begin{proof}
It is easy to check that the codes in the statement of the theorem are dually qOACs. We now prove that, up to equivalence, they are the only ones. 
We start by analyzing the case when 
$\mathcal{C}\supseteq\Mat(U)$, for some $U\subseteq\F_q^n$ of $\dim(U)=\alpha$. Up to equivalence, we may assume that $U=\langle e_1,\ldots,e_\alpha\rangle$. Write $\mathcal{C}=\Mat(\langle e_1,\ldots,e_\alpha\rangle)+\langle M_1,\ldots,M_\rho\rangle$, where $M_1,\ldots,M_\rho\in\Mat(\langle e_1,\ldots,e_\alpha\rangle)^\perp=\Mat(\langle e_{\alpha+1},\ldots,e_n\rangle)$ are linearly independent. Since $\mathcal{C}$ is a qOAC, then $\maxrk(\mathcal{C})=\alpha+1$. 
We claim that $\maxrk(\langle M_1,\ldots,M_\rho\rangle)=1$. In fact, any $M\in \langle M_1,\ldots,M_\rho\rangle\subseteq\Mat(\langle e_{\alpha+1},\ldots,e_n\rangle)$ has $$\dim(\rowsp(M))=\rank(M)\leq\maxrk(\Mat(\langle e_{\alpha+1},\ldots,e_n\rangle))=n-\alpha.$$
Let $L\in\Mat(\langle e_1,\ldots,e_\alpha\rangle)$ be a matrix whose first $\alpha$ rows are linearly independent vectors in a vector space $V$ such that $V\oplus\rowsp(M)=\mathbb{F}_q^m$. Notice that one can always find such an $L$, since $\dim(V)=m-\dim(\rowsp(M))\geq m-(n-\alpha)\geq\alpha$. Then $\rowsp(L)\cap\rowsp(M)=0$, so $L+M\in \mathcal{C}$ has 
$$\rank(L + M)=\rank(M)+\rank(L)\leq\maxrk(\mathcal{C})=\alpha + 1,$$ 
which proves that $\rank(M)\leq 1$, since $\rank(L)=\alpha$.
Since $\maxrk(\langle M_1,\ldots,M_\rho\rangle)=1$, then either 
$\langle M_1,\ldots,M_\rho\rangle\subseteq\Mat(w)$ for some $w\in\langle e_{\alpha+1},\ldots,e_{n}\rangle\subseteq \mathbb{F}_q^n$, or 
$\langle M_1,\ldots,M_\rho\rangle\subseteq\Mat(w)^t$ for some $w\in\mathbb{F}_q^m$. Since $\langle M_1,\ldots,M_\rho\rangle\subseteq\Mat(\langle e_{\alpha+1},\ldots,e_n\rangle)$, the latter is only possible if $$\rho=\dim(\langle M_1,\ldots,M_\rho\rangle)\leq \dim(\Mat(\langle e_{\alpha+1},\ldots,e_n\rangle)\cap\Mat(w)^t)=n-\alpha.$$

If $\langle M_1,\ldots,M_\rho\rangle\subseteq\Mat(w)$, then, after suitable invertible operations involving the last $n-\alpha$ rows, we may suppose that $w=e_{\alpha+1}$. Moreover, after suitable invertible column operations
$$M_i = \begin{pmatrix}
0_{n\times(i-1)} & e_{\alpha+1} & 0_{n\times(m-i)} \\
\end{pmatrix} \quad \mbox{ for } 1\leq i \leq \rho.$$
Since both types of operations fix $\Mat(\langle e_1,\ldots,e_\alpha\rangle)$, we have shown that
\begin{equation} \label{dqOAC1}
\mathcal{C} \sim\left\{  \begin{pmatrix}
A & B\\
u & 0_{1 \times (m-\rho)}  \\
0_{(n-\alpha -1) \times \rho} & 0_{(n-\alpha-1) \times (m-\rho)}\\
\end{pmatrix} : A \in \F_q^{\alpha \times \rho}, B \in \F_q^{\alpha \times (m-\rho)}, u \in \F_q^{1\times\rho}\right\}.
\end{equation}
This yields the codes in part (a) of the statement.

If 
$\rho\leq n-\alpha$ and $\langle M_1,\ldots,M_\rho\rangle\subseteq\Mat(w)^t$, then, after suitable invertible operations involving the last $n-\alpha$ rows, we may suppose that $M_i$ is the matrix whose rows are all zero, except for row $\alpha+i$ which is equal to $w$. Up to invertible column operations, we may further suppose that $w=e_1\in\mathbb{F}_q^m$. Since both types of operations fix $\Mat(\langle e_1,\ldots,e_\alpha\rangle)$, we have that
$$\mathcal{C} \sim \left\{ \begin{pmatrix}
u & A \\
v & 0_{\rho\times(m-1)}\\
0_{(n-\alpha-\rho)\times 1} & 0_{(n-\alpha-\rho)\times(m-1)}\\
\end{pmatrix} : A\in \F_q^{\alpha\times(m-1)}, \\ 
u\in \F_q^{\alpha\times 1}, v\in \F_q^{\rho\times 1}\right\}.$$
This yields the codes in part (b) of the statement. 

Suppose now that $\C\subseteq\Mat(U)$ for some $U\subseteq\F_q^n$ of $\dim(U)=\alpha+1$. Up to equivalence, we may assume that $U=\Mat(\langle e_{n-\alpha},\ldots,e_n\rangle)$, that is,  $\mathcal{C}^\perp\supseteq\Mat(\langle e_1,\ldots,e_{n-\alpha-1}\rangle)$. Then the above argument shows that either
$$\mathcal{C}^\perp \sim\left\{  \begin{pmatrix}
A & B\\
u & 0_{1\times\rho}  \\
0_{\alpha\times(m-\rho)} & 0_{\alpha\times \rho}\\
\end{pmatrix} : A \in \F_q^{(n-\alpha-1)\times(m-\rho)}, B \in \F_q^{(n-\alpha-1) \times\rho}, u \in \F_q^{1\times (m-\rho)}\right\}$$ or
$$\mathcal{C}^\perp \sim \left\{ \begin{pmatrix}
u & A \\
v & 0_{(m-\rho)\times(m-1)}\\
0_{(\alpha+\rho+1-m)\times 1} & 0_{(\alpha+\rho+1-m)\times(m-1)}\\
\end{pmatrix} : A\in \F_q^{(n-\alpha-1)\times(m-1)}, \\ 
u\in \F_q^{(n-\alpha-1)\times 1}, v\in \F_q^{(m-\rho)\times 1}\right\}.$$
Moreover, the latter is only possible when $\rho\geq m-\alpha-1$. Taking duals, we obtain that either $\mathcal{C}$ has the form (\ref{dqOAC1}), or
$$\mathcal{C}\sim\left\{\begin{pmatrix}
A & u \\
B & 0_{(m-\rho)\times 1}\\
0_{(n-\alpha-1)\times (m-1)} & 0_{(n-\alpha-1)\times 1}\\
\end{pmatrix} : A\in \F_q^{(\alpha+\rho+1-m)\times(m-1)}, \\ 
B\in \F_q^{(m-\rho)\times(m-1)}, u\in \F_q^{(\alpha+\rho+1-m)\times 1}\right\}.$$
This yields the codes in part (c) of the statement.

Finally, suppose that $\mathcal{C}\not\supseteq\Mat(U)$ for any $U\subseteq\F_q^n$ of $\dim(U)=\alpha$ and  $\mathcal{C}\not\subseteq\Mat(U)$ for any $U\subseteq\F_q^n$ of $\dim(U)=\alpha+1$. This implies also that $\mathcal{C}^\perp\not\subseteq\Mat(U)$ for any $U\subseteq\F_q^n$ of $\dim(U)=n-\alpha$. Since $n\leq m$, then either $\rho\leq m-\alpha-1$ or $\rho\geq n-\alpha$. This implies that, for any $0<\rho<m$, Theorem~\ref{firstclasstheorem} applies to either $\mathcal{C}$ or $\mathcal{C}^\perp$. Because of our assumptions, and since the code of Theorem~\ref{firstclasstheorem} (b) (iv) is not a dually qOAC, $\C$ or $\C^\perp$ is a code as in Theorem~\ref{firstclasstheorem} (b) (iii). 
Therefore $m=n+1$ and $\rho=m-\alpha-1=n-\alpha$ and we obtain the codes in part (d) of the statement.
\end{proof}

Theorem \ref{firstclasstheorem} and Theorem \ref{corollaryC} provide us with many examples of codes which are qOACs but not dually qOACs. In fact, any code as in Theorem \ref{firstclasstheorem} which is not one of the codes in Theorem \ref{corollaryC} is of this kind. We now give some concrete examples of qOACs which are not dually qOACs. These examples are covered by Theorem 2.5, unless $m=n$ and $\rho=n-\alpha-1$.

\begin{example}\label{notdqOAC}
Let $1\leq\alpha\leq n-1$ and $m-\alpha-1\leq\rho\leq m-2$. Let
$$ \begin{array}{ccl} 
\C &=& \Bigg \{ \begin{pmatrix}
\mathcal{V} & B \\
A & C \\
0_{(n-\alpha-1) \times (m-\rho)} & 0_{(n-\alpha-1) \times \rho} \\
\end{pmatrix} : \mathcal{V} \subseteq \F_q^{(m-\rho) \times (m-\rho)} \mbox{ the set of matrices} \\ 
& & \\
& & \mbox{ with 0 on the diagonal}, A \in \F_q^{(\alpha+1-m+\rho) \times (m-\rho)}, 
 B \in \F_q^{(m-\rho) \times \rho}, C \in \F_q^{(\alpha + 1 - m +\rho) \times \rho }\Bigg \}. \end{array}$$
 Then $\C \subseteq \mat$ has dimension $\alpha m + \rho$ and maximum rank $\alpha + 1$, hence it is a qOAC. Its dual is given by 
 $$ \begin{array}{ccc} 
\Cd &=& \Bigg \{ \begin{pmatrix}
\mathcal{D} & 0_{(m-\rho)\times\rho} \\
0_{(\alpha+1-m+\rho)\times(m-\rho)} & 0_{(\alpha+1-m+\rho)\times \rho} \\
A & B\\
\end{pmatrix} : A \in \F_q^{(n-\alpha-1) \times (m-\rho)}, 
 B \in \F_q^{(n-\alpha-1) \times \rho}, \\ 
& & \\
& & \mathcal{D} \subseteq \F_q^{(m-\rho) \times (m-\rho)} \mbox{ the set of diagonal matrices} \Bigg \}.\end{array}$$
$\C^\perp$ has dimension $(n-\alpha)m - \rho$ and maximum rank
$m+n-\rho-\alpha-1\geq n-\alpha+1$. Hence $\Cd$ is not a qOAC, that is, $\C$ is not a dually qOAC.
\end{example}

Another result from \cite{Seguins2} provides us with information on qOACs beyond those treated in Theorem \ref{firstclasstheorem}. Below we state the result in our language.

\begin{theorem}[\cite{Seguins2}, Theorem 1.7]\label{secondclasstheorem}
Let $\mathcal{C} \subseteq \mat$ be a qOAC with $\dim(\mathcal{C}) = \alpha m + \rho$, $0<\rho<m$. If $q = 2$ suppose that $\rho \geq 2(n-\alpha+1) - m$, else suppose that $\rho \geq 2(n - \alpha) - m$. Then $\mathcal{C}$ is equivalent to a linear subspace of  $\Mat(\langle e_1,\ldots,e_s\rangle)+\Mat(\langle e_1,\ldots,e_{\alpha+1-s}\rangle)^t$, for some $s\in\{0,\ldots,\alpha+1\}$. 
\end{theorem}

If $\C$ is a qOAC of dimension $\dim(\mathcal{C}) = \alpha m + \rho$ and $\maxrk(\C)=\alpha+1$, then $\C$ cannot be equivalent to a subspace of $\Mat(\langle e_1,\ldots,e_s\rangle)+\Mat(\langle e_1,\ldots,e_k\rangle)^t$ for any $s,k\geq 0$ with $s+k\leq\alpha$. 
Lemma 1 in \cite{FLA} shows that, up to equivalence, $\C$ is contained in $\Mat(\langle e_1,\ldots,e_{\alpha+1}\rangle)+\Mat(\langle e_1,\ldots,e_{\alpha+1}\rangle)^t$. 
Theorem \ref{firstclasstheorem} and Theorem \ref{secondclasstheorem} prove that, under some assumptions on $\rho$, $\C$ is equivalent to a subspace of $\Mat(\langle e_1,\ldots,e_s\rangle)+\Mat(\langle e_1,\ldots,e_{\alpha+1-s}\rangle)^t$. In the next example, we exhibit a code which does not satisfy the assumptions of Theorem \ref{firstclasstheorem} or Theorem \ref{secondclasstheorem} and which is not equivalent to a subspace of $\Mat(\langle e_1,\ldots,e_s\rangle)+\Mat(\langle e_1,\ldots,e_{\alpha+1-s}\rangle)^t$ for any $0\leq s\leq \alpha+1$. The example is based on \cite[Remark on pg. 228]{MES}. 

\begin{example}
 Let $\C \subseteq \F_2^{4 \times 4}$ be a linear code given by 
$$\C = \left\{ \begin{pmatrix}
a_1 & a_4 & a_5 & a_6\\
0 & a_2 & a_7 & a_8 \\
0 & 0 & a_2 + a_3 & a_9\\
0 & 0 & 0 & a_3\\
\end{pmatrix} : a_i \in \F_2 \mbox{ for } 1 \leq i \leq 9 \right\}.$$
Then $\maxrk(\C) = 3$ and $\dim(\C) = 9$, hence $\C$ is a qOAC. 
We claim that $\C$ is not equivalent to a subspace of $\Mat(\langle e_1,\ldots,e_s\rangle)+\Mat(\langle e_1,\ldots,e_{3-s}\rangle)^t$ for any $0\leq s\leq 3$. To see this, consider the following two elements of $\C$:
$$M_1 = \begin{pmatrix}
1 & 0 & 0 & 0\\
0 & 1 & 0 & 0\\
0 & 0 & 1 & 0 \\
0 & 0 & 0 & 0 \\
\end{pmatrix} \mbox{ and } M_2 = \begin{pmatrix}
0 & 0 & 0 & 0\\
0 & 0 & 0 & 0\\
0 & 0 & 1 & 0\\
0 & 0 & 0 & 1\\
\end{pmatrix}. $$
We prove that if $M_1$ is contained in $\Mat(U) + \Mat(V)^t$ for some $U,V \subseteq \F_2^4  $ of $\dim(V) + \dim(U) = 3$, then $U,V\subseteq \langle e_1, e_2, e_3\rangle$. Notice that this establishes the claim, since $M_2$ is not contained in such a space. 

The only non trivial case, up to transposition, is $\dim(U) = 2$ and  $\dim(V) = 1$. Let $U = \langle (x_1, y_1, z_1, 0) , (x_2, y_2, z_2, w) \rangle \subseteq \F_2^4 $ and $V = \langle (a,b,c,d)\rangle   \subseteq \F_2^4$. By assumption there are $A \in \Mat(\langle (x_1, y_1, z_1, 0) , (x_2, y_2, z_2, w) \rangle)$ and $B \in \Mat(\langle (a,b,c,d)\rangle)^t$, such that $M_1 = A + B$. Equating the last column of $A+B$ to that of $M_1$ yields the following homogeneous linear system:
\begin{equation} \label{GLS}
\begin{array}{ccccccc}
0 & = & \alpha d &+& t_1 x_1 &+& t_2 x_2 \\
0 & = & \beta d &+& t_1 y_1 &+& t_2 y_2 \\
0 & = & \gamma d &+& t_1 z_1 &+& t_2 z_2 \\
0 & = & \delta d &+& 0 &+& t_2 w. \\
\end{array}
\end{equation}
Notice that $(\alpha, \beta , \gamma, \delta ) \notin U$, otherwise $M_1$ would be contained in $\Mat(U)$, but this is not possible since $\maxrk(\Mat(U))=2$. Therefore, the only solution of (\ref{GLS}) regarded as a system in the variables $d,t_1,t_2$ is $d=t_1=t_2=0$.    

Applying the same argument to the last row of $M_1$ yields another homogeneous linear system, which, regarded as system in the variables $\delta,w$ only has the trivial solution $\delta=w=0$. In particular, we have proved that $d=w=0$, i.e., $U,V\subseteq\langle e_1, e_2, e_3 \rangle.$
\end{example}

In the sequel, we concentrate on linear spaces of the form $\Mat(\langle e_1,\ldots,e_s\rangle)+\Mat(\langle e_1,\ldots,e_k\rangle)^t$, as well as some of their linear subspaces. 

\begin{definition}\label{specialcodes}
Let $s,h,k\geq 0$ be integers such that $k\leq m$ and $0<s+h\leq n$.
Let $$\mathcal{C}_{s,h,k}=\left\{ \begin{pmatrix}
A & B  \\
C & 0_{h \times (m-k)} \\
0_{(n-s-h) \times k} & 0_{(n-s-h) \times (m-k)} \\
\end{pmatrix} : A \in \F_q^{s \times k} , B \in \F_q^{s \times (m-k)} , C \in \F_q^{h \times k}, \right\}\subseteq\mat.$$
\end{definition}

\begin{remark}
Up to equivalence, all the dually qOAC are of the form $\mathcal{C}_{s,h,k}$ for some $s,h,k$ by Theorem \ref{corollaryC}.
Moreover, one can check that $C_{s,h,k}$ is a dually qOAC if and only if there are parameters $0\leq\alpha<n$ and $0<\rho<m$ such that $(s,h,k)$ is among $(\alpha,\rho,1), (\alpha,1,\rho), (0,n,\alpha+1)$ if $m=n+1$, or $(\alpha+\rho+1-m,m- \rho,m-1)$ if $\rho\geq m-\alpha-1$.
\end{remark}

In the rest of the paper, we compute the invariants of the codes from Definition \ref{specialcodes}. In the next proposition, we characterize which codes of the form $\mathcal{C}_{s,h,k}$ are qOACs. 
Together with Theorem~\ref{corollaryC}, Proposition~\ref{structure} yields examples of qOACs which are not dually qOACs, beyond those of Example~\ref{notdqOAC}. Some examples of this kind are given in Example \ref{newqOAC}.

\begin{proposition}\label{structure}
Let $s,h,k$ be non-negative integers such that $0<s+h \leq n$ and $k<m$.
Then $\mathcal{C}_{s,h,k}$ is a qOAC if and only if 
$$0<\min\{h,k\}\leq\left\lfloor\frac{m-1}{m-\max\{h,k\}}\right\rfloor.$$
\end{proposition}

\begin{proof}
The code $\mathcal{C}$ has dimension $\dim(\mathcal{C})=sm+hk$ and maximum rank $s+\min\{ h,k\}$. Then $\mathcal{C}$ is a qOAC iff $m\nmid hk$ and
$$\left\lceil \frac{sm + hk}{ m}\right\rceil = s+\min\{h,k\}  \quad \Longleftrightarrow \quad \left\lceil \frac{hk}{ m}  \right\rceil = \min\{h,k\}. $$
Therefore, $\mathcal{C}$ is a qOAC iff $\min\{h,k\} >0$ and $(\min\{h,k\}-1)m+1 \leq hk \leq \min\{h,k\}m-1$. A computation shows that $\mathcal{C}$ is a qOAC iff 
$$0<\min\{h,k\} \leq\left\lfloor \frac{m-1}{m-\max\{h,k\}} \right\rfloor.$$ Notice that $\max\{h,k\}<m$, since $m\nmid hk$.
\end{proof}

\begin{example}\label{newqOAC}
Let $0<k<m$ and let
$$\mathcal{C}=\left\{ \begin{pmatrix}
A & 0_{n \times (m-k)} \\
\end{pmatrix} : A\in \F_q^{n\times k}, \right\}\subseteq\mat.$$
By Theorem~\ref{corollaryC}, $\mathcal{C}$ is a dually qOAC if and only if $m=n+1$.
By Proposition~\ref{structure}, $\mathcal{C}$ is a qOAC if and only if 
$$k\leq\min\left\{\left\lfloor\frac{m-1}{m-n}\right\rfloor, n\right\}\ \mbox{ or }\ n\leq\min\left\{\left\lfloor\frac{m-1}{m-k}\right\rfloor, k\right\}.$$
\end{example}

In the following sections, we compute the invariants of the codes from Definition \ref{specialcodes}. Since this does not affect the computation of the invariants, we always assume that $n=s+h$. This amounts to considering the codes $C_{s,k,n-s}=\Mat(\langle e_1,\ldots,e_s\rangle)+\Mat(\langle e_1,\ldots,e_k\rangle)^t$.

\section{Generalized weights}\label{Secgenweights}

In this section we compute the generalized weights of codes of the form $\Mat(\langle e_1,\ldots,e_s\rangle)+\Mat(\langle e_1,\ldots,e_k\rangle)^t$. For what discussed in the previous section, this determines the weights of all dually qOACs and of certain qOACs. We start by recalling the definition of generalized weights.

\begin{definition}\label{defgw}(\cite{ACapproach}, Definition~23)
Let $\mathcal{C} \subseteq \mat$ be a rank-metric code. For $1\leq i\leq\dim(\C)$, the $i$th \textbf{generalized weight} of $\mathcal{C}$ is
$$a_i(\mathcal{C}) = \displaystyle \frac{1}{m} \min  \{ \dim(\mathcal{A}) \, : \, \mathcal{A} \subseteq \mat \mbox{ is an optimal anticode with } \dim(\mathcal{A} \cap \mathcal{C}) \geq i \}.$$
\end{definition}

Generalized weights were defined in \cite{ACapproach}, where some of their basic properties were also established. A different - but related - definition of generalized weights was given in \cite{MMP}.

We start with a simple result, which will be used in the proof of Theorem \ref{theorem}.

\begin{proposition} \label{minmax}
Let $\mathcal{C}\subseteq \mat$ be a rank-metric code which contains an optimal anticode  of dimension $sm$ for some $1 \leq s \leq n$. Then 
$$a_{(i-1)m + 1}(\mathcal{C})=\cdots=a_{i m}(\mathcal{C})=i \mbox{ for } 1 \leq i \leq s.$$
\end{proposition}
 
\begin{proof}
Fix $0\leq i\leq s$ and let $\A$ be an optimal anticode of dimension $\dim(\A)=im$ such that $\A \subseteq \mathcal{C}$. Then $a_{im}(\mathcal{C}) = i$, since $a_{im}(\mathcal{C}) \geq i$ by \cite[Theorem 30]{ACapproach} and $a_{im}(\C) \leq i$ since $\C\supseteq\A$. Using the properties of generalized weights from \cite[Theorem 30]{ACapproach}, one easily obtains that
$$a_{(i-1)m + 1}(\mathcal{C}) = \cdots = a_{im}(\mathcal{C}) = i.$$
\end{proof}

We are now ready to state the main theorem of this section. The next theorem computes the generalized weights of all dually qOACs and some qOACs.

\begin{theorem} \label{theorem}
Let $\mathcal{C}=\Mat(\langle e_1,\ldots,e_s\rangle)+\Mat(\langle e_1,\ldots,e_k\rangle)^t\subseteq\mat$. Then
$$\begin{array}{clcll}
a_{(i-1)m + 1}(\mathcal{C})& =\cdots=& a_{i m}(\mathcal{C})& =i & \mbox{ for } 1 \leq i \leq s, \\
a_{sm + (i-1)k+1}(\mathcal{C})& =\cdots=& a_{sm + ik}(\mathcal{C}) & =s+i & \mbox{ for } 1 \leq i \leq n-s.
\end{array}$$
\end{theorem}

\begin{proof}
Since $\mathcal{C}\supseteq \Mat(\langle e_1,\ldots,e_s\rangle)$, then by Proposition \ref{minmax} 
$$a_{(i-1)m + 1}(\mathcal{C}) = \cdots = a_{i m }(\mathcal{C}) = i \quad \mbox{ for } 1 \leq i \leq s.$$
Fix $1 \leq i \leq n-s$ and let $\A=\Mat(V)\subseteq \mat$ be an optimal anticode with $V\subseteq\F_q^n$ of $\dim(V)=s+i$. We claim that $\dim(\C\cap\A)\leq sm+ik$.
Since $\dim(V)=s+i$, then there exist $i$ linearly independent vectors $v_1,\ldots,v_i\in V\setminus\langle e_1,\ldots,e_s\rangle$. Consider the matrices
$$M_j^1 = \begin{pmatrix}
0_{n \times k} & v_j & 0_{n \times m -k-1} \\
\end{pmatrix}, \ M_j^2 = \begin{pmatrix}
0_{n \times k + 1} & v_j & 0_{n \times m -k-2} \\
\end{pmatrix}, \ \ldots \ , M_j^{m-k} = \begin{pmatrix}
0_{n \times m-1} & v_j \\
\end{pmatrix}$$ 
for $ 1 \leq j \leq i$. Clearly 
$\langle M_1^1,\ldots,M_i^{1},\cdots,M_1^{m-k},\ldots,M_i^{m-k}\rangle\subseteq\A$. Moreover
\begin{equation} \label{11.1}
\mathcal{C} \ \cap \ \langle M_1^1 , \ldots , M_i^{1}, \cdots , M_1^{m-k}, \ldots , M_i^{m-k} \rangle =0.
\end{equation}
Hence 
$$ \begin{array}{ccl}  
\dim(\mathcal{C}) + \dim(\A) - \dim(\mathcal{C} \cap \A)& = & \dim(\mathcal{C} + \A) \\
&\geq & \dim(\mathcal{C}) + \dim (\langle M_1^1 , \ldots , M_i^{1}, \cdots , M_1^{m-k}, \ldots , M_i^{m-k}\rangle ),\\
\end{array}
$$
which proves our claim that
$$\dim(\mathcal{C} \cap \A) \leq sm + ik.$$
It follows directly that
\begin{equation} \label{11.2}
a_{sm+ik+1}(\mathcal{C})\geq s+i+1 \mbox{ for } 1\leq i\leq n-s-1.
\end{equation}
The inequality \begin{equation} \label{11.3}
a_{sm+1}(\C)\geq s+1\end{equation} 
follows from \cite[Theorem 30]{ACapproach}.

In order to prove that all the previous inequalities are in fact equalities, consider $\A=\Mat(\langle e_1,\ldots,e_{s+i}\rangle)$. It is easy to check that $\dim(\mathcal{C}\cap\A)=sm+ik $. 
Therefore $a_{sm +ik}(\mathcal{C})\leq s+i$ which, together with (\ref{11.2}) and (\ref{11.3}) yields
$$a_{sm+(i-1)k+1}(\mathcal{C})=\cdots=a_{sm+ik}(\mathcal{C})=s+i+1,$$
for $1\leq i\leq n-s$.
\end{proof}

\section{Rank distribution}\label{SecRankdistr}

In this section we compute the rank distribution of codes of the form $\Mat(\langle e_1,\ldots,e_s\rangle)+\Mat(\langle e_1,\ldots,e_k\rangle)^t$. For what discussed in Section \ref{SecqOACs}, this determines the rank distribution of all dually qOACs and of certain qOACs. We start by recalling a basic result from linear algebra.

\begin{definition}
The {\bf Gaussian $q$-binomial coefficient} is the integer $$\binom{n}{r}_q = \frac{(q^n - 1)(q^n-q)\cdots (q^n-q^{r-1})}{(q^r - 1)(q^r-q)\cdots (q^r-q^{r-1})}.$$
\end{definition}

\begin{lemma}  \label{rank_k}
The number of rank $r$ matrices in $\mat$ is 
$$\phi_q(n,m,r)=\binom{n}{r}_q (q^m -1)(q^m-q)\cdots (q^m - q^{r-1}).$$
\end{lemma}

\begin{theorem}
Let $\mathcal{C}=\Mat(\langle e_1,\ldots,e_s\rangle)+\Mat(\langle e_1,\ldots,e_k\rangle)^t$. For $0\leq r\leq n$ denote by $A_r$ the number of elements of rank $r$ in $\C$. Then 
$$A_r = \phi_q(s,m,r) + \sum_{i=1}^{\min\{n-s,k,r\}} \binom{k}{i}_q  \binom{m-i}{r-i}_q\   q^{si}\  \prod_{t=0}^{i-1}(q^{n-s} - q^t)\prod_{j=0}^{r-i-1}  (q^{s} - q^j),$$
for $1\leq r \leq \min\{s+k,n\}$. 
\end{theorem}

\begin{proof}
The number of $n\times m$ matrices in $\C$ of rank $1\leq r\leq n$ is the sum of the number of matrices of rank $r$ in $\F_q^{s\times m}$ and the number of $n\times m$ matrices of rank $r$ in $\mat$ which have at least one nonzero row among the last $n-s$.  We denote the latter set of matrices by $\mathcal{C}_r$. 

The first number is $\phi_q(s,m,r)$ by Lemma \ref{rank_k}. We now compute the cardinality of $\mathcal{C}_r$.
Let $1 \leq i \leq \min\{n-s,k,r\}$ and let $V\subseteq \F_q^m$ of $\dim(V)=i$. 
Let $U \subseteq \F_q^m$ be an $r$-dimensional vector space containing $V$. Since $r$-dimensional vector spaces $U$ containing a given $i$-dimensional subspace are in one-to-one correspondence with $(r-i)$-dimensional vector spaces in $\F_q^{m-i}$, the number of $r$-dimensional vector spaces $U$ such that of $\F_q^m\supseteq U\supseteq V$ is 
\begin{equation} \label{rowspaces}
\binom{m-i}{r-i}_q.
\end{equation}

Next we determine the number of matrices with row space $U$ and such that the row space of the last $n-s$ rows is $V$. Complete a basis $u_1, \ldots , u_i$ of $V$ to a basis $u_1, \ldots , u_r$ of $U$.
We observe that every such $M \in \mathcal{C}_r$ is of the form
$$  M = D \cdot \begin{pmatrix}
u_1 \\
\vdots \\
u_r \\
\end{pmatrix} \in \mathcal{C}_r, \mbox{ where } D=\begin{pmatrix}
A & B  \\
C & 0_{(n-s) \times (r-i)} \\
\end{pmatrix},$$
$A \in \F_q^{s \times i}, B \in \F_q^{s \times (r-i)}, C \in \F_q^{(n-s)\times i}$. Moreover $B,C$, and $D$ must have full rank. Since there is a one-to-one correspondence between $M$ and $D$, the number of  matrices in $\mathcal{C}_r$ with row space $U$ and such that the last $n-s$ rows generate $V$ is
\begin{equation}\label{fullrankA}
 q^{si} \prod_{t=0}^{i-1} (q^{n-s} - q^t) \prod_{j=0}^{r-i-1} (q^s - q^j).
\end{equation}
The product of (\ref{rowspaces}) and (\ref{fullrankA}) is then the number of matrices in $\mathcal{C}_r$, whose last $n-s$ rows span $V$.

Finally, observe that for a given $1 \leq i \leq \min\{n-s,k,r\}$ there are
$$\binom{k}{i}_q$$ 
$i$-dimensional subspaces $V$ in $\F_q^{k}$. Hence the cardinality  of $\C_r$ is
$$\sum_{i=1}^{\min\{n-s,k,r\}} \binom{k}{i}_q  \binom{m-i}{r-i}_q   \, q^{si} \prod_{t=0}^{i-1}  (q^{h} - q^t) \prod_{j=0}^{r-i-1}  (q^{s} - q^j).$$ 
\end{proof}

\section{Associated $q$-Polymatroids}\label{SecPolymat}

$q$-polymatroids are the $q$-analogs of polymatroids. They were introduced indipendently by Shiromoto in \cite{Sh} and by Gorla, Jurrius, L\'opez Valdez, and Ravagnani in \cite{q-poly}. 
In this section we compute the rank functions of the $q$-polymatroids associated to codes of the form $\Mat(\langle e_1,\ldots,e_s\rangle)+\Mat(\langle e_1,\ldots,e_k\rangle)^t$. By Lemma \ref{lemma:reduction} this determines the $q$-polymatroids associated to all dually qOACs and to certain qOACs. 

We start by recalling the relevant definitions. The fact that the functions $\rho_c$ and $\rho_r$ define $q$-polymatroids is shown in \cite[Theorem 5.3]{q-poly}.

\begin{definition}\label{defpoly}
Let $\C\subseteq\mat$ be a rank-metric code, and let $J \subseteq \F_q^n$ and $K \subseteq \F_q^m$ be linear subspaces. The {\bf $q$-polymatroids} associated to $\C$
are $(\F_q^n,\rho_c)$ and $(\F_q^m,\rho_r)$, where
$$\rho_c(\mathcal{C},J)=
\frac{\dim(\mathcal{C})-\dim(\mathcal{C}\cap\Mat(J^\perp))}{m}$$ and
$$\rho_r(\mathcal{C},K)=
\frac{\dim(\mathcal{C})-\dim(\mathcal{C}\cap\Mat(K^\perp)^t)}{n}.$$
\end{definition}

The next proposition follows directly from the definition of the rank functions.

\begin{proposition}
Let $\mathcal{C}\subseteq\mat$ be a rank-metric code and assume that $m\nmid \dim(\mathcal{C})$. The following are equivalent:
\begin{enumerate}
\item $\mathcal{C}$ is a qOAC,
\item $\lceil \rho_c(\mathcal{C}, \F_q^n) \rceil = \maxrk(\mathcal{C})$,
\item $m=n$ and $\lceil\rho_r(\mathcal{C},\F_q^m)\rceil=\maxrk(\mathcal{C})$.
\end{enumerate}
\end{proposition}

The next lemmas will be used in the proof of Theorem \ref{thq-poly}.

\begin{lemma}\label{Ccapw}
Let $h\geq 0$, $0\leq k\leq m$, and $0\leq s\leq m-h$.
Let $\mathcal{C}=\Mat(\langle e_1,\ldots,e_s\rangle)+\Mat(\langle e_1,\ldots,e_k\rangle)^t\subseteq\F_q^{(s+h)\times m}$ be a rank-metric code. Let $J\subseteq\F_q^{s+h}$ and $K\subseteq\F_q^m$. Then:
\begin{itemize}
\item[(a)] If $J\cap\langle e_1,\ldots,e_s\rangle=0$, then
$\dim(\mathcal{C} \cap \Mat(J))=k\cdot\dim(J).$
\item[(b)] If $K\cap\langle e_1,\ldots,e_k\rangle=0$, then
$\dim(\mathcal{C}\cap\Mat(K)^t)=s\cdot\dim(K).$
\end{itemize}
\end{lemma}

\begin{proof}
We only prove the first statement, as the second is proved similarly. Let $M\in\mathcal{C}\cap \Mat(J)$. The last $m-k$ columns of $M$ belong to $J\cap \langle e_1,\ldots,e_s\rangle$, hence they are zero. Therefore $\mathcal{C}\cap\Mat(J)\subseteq\Mat(\langle e_1,\ldots,e_k\rangle)^t\cap\Mat(J)\subseteq \mathcal{C}\cap\Mat(J),$
where the second inclusion follows from $\mathcal{C}\supseteq\Mat(\langle e_1,\ldots,e_k\rangle)^t$.
It follows that $$\mathcal{C}\cap\Mat(J)=\Mat(\langle e_1,\ldots,e_k\rangle)^t\cap\Mat(J),$$
in particular $\dim(\mathcal{C} \cap \Mat(J))=k\cdot\dim(J)$. 
\end{proof}

The proof of the next lemma is immediate.

\begin{lemma}\label{lemma:reduction}
Let $V\subseteq\F_q^n$ and let $\C\subseteq\Mat(V)\subseteq\mat$ be a rank-metric code. Then $$\C\cap\Mat(J)=\C\cap\Mat(J\cap V)$$ for any $J\subseteq\F_q^n$. In particular, if $\C\subseteq\Mat(\langle e_1,\ldots,e_\ell\rangle)$ for some $\ell\leq n$, then one can regard $\C$ as a rank-metric subcode of $\F_q^{\ell\times m}$ by deleting the last $n-\ell$ rows of each matrix. Moreover, the $q$-polymatroid $(\F_q^m,\rho_r)$ is left unchanged by this operation and the $q$-polymatroid $(\F_q^\ell,\rho_c)$ determines the $q$-polymatroid $(\F_q^n,\rho_c)$ according to the formula above. 
\end{lemma}

The next theorem is the main result of this section. We compute the rank functions of the $q$-polymatroids associated to a quasi-optimal anticode.

\begin{theorem} \label{thq-poly}
Let $h\geq 0$, $0\leq k\leq m$, and $0\leq s\leq m-h$.
Let $V=\langle e_1,\ldots,e_s\rangle, V'=\langle e_1,\ldots,e_{s+h}\rangle\subseteq\F_q^n$ and let $U=\langle e_1,\ldots,e_k\rangle\subseteq\F_q^m$. Let $\mathcal{C}=\Mat(V)+(\Mat(U)^t\cap\Mat(V'))\subseteq\mat$ be a rank-metric code. Let $J\subseteq\F_q^n$ and $K\subseteq\F_q^m$. 
The rank functions of the $q$-polymatroids associated to $\mathcal{C}$ are
$$\rho_r(\mathcal{C},K)=\frac{h(k-\dim(U\cap K^\perp)) +s\cdot\dim(K)}{n}$$
and 
$$\rho_c(\mathcal{C},J)=s-\dim(V\cap J^\perp)+\frac{k(h+\dim(V\cap J^\perp)-\dim(V'\cap J^\perp))}{m}.$$
\end{theorem}

\begin{proof}
For any $J\subseteq\F_q^n$, one has \begin{equation}\label{eq0}
\C\cap\Mat(J)=\C\cap\Mat(J\cap V')
\end{equation} by Lemma \ref{lemma:reduction}, since $\C\subseteq\Mat(V')$. Write $J\cap V'$ as $(J\cap V)+J'$, where $J'\cap V=0$. This can always be done by letting $J'$ be the vector space generated by a set of vectors that, together with a basis of $J\cap V$, form a basis of $J\cap V'$. Then \begin{equation}\label{eq1}
\C\cap\Mat(J\cap V')=\C\cap\Mat(J\cap V+J')=\C\cap[\Mat(J\cap V)+\Mat(J')]=\Mat(J\cap V)+\C\cap\Mat(J').
\end{equation} Since $J'\cap V=0$, then 
\begin{equation}\label{eq2}
\dim(\Mat(J\cap V)+\C\cap\Mat(J'))=\dim(\Mat(J\cap V))+\dim(\C\cap\Mat(J')).
\end{equation}
Moreover $\dim(\C\cap\Mat(J'))=k\cdot\dim(J')$ by Lemma \ref{Ccapw}.
Combining (\ref{eq0}), (\ref{eq1}), and (\ref{eq2}) one gets $$\dim(\C\cap\Mat(J))=m\cdot\dim(J\cap V)+k\cdot\dim(J')=m\cdot\dim(J\cap V)+k(\dim(J\cap V')-\dim(J\cap V)).$$
Therefore $$\rho_c(\mathcal{C},J)=\frac{\dim(\C)-\dim(\C\cap\Mat(J^\perp))}{m}=s-\dim(V\cap J^\perp)+\frac{k(h+\dim(V\cap J^\perp)-\dim(V'\cap J^\perp))}{m}$$ as claimed.
The other equality is proved similarly.
\end{proof}

\end{document}